\documentclass[10pt]{article}

\usepackage{hyperref}
\usepackage{comment}
\usepackage{tikz}
\usetikzlibrary{shapes,calc,arrows,trees,positioning,through,intersections,fadings}
\def\nottoobig#1{{\hbox{$\left#1\vcenter to1.111\ht\strutbox{}\right.\n@space$}}}
\usepackage{color}
\usepackage{soul}
\usepackage{graphicx}
\definecolor{lightblue}{rgb}{.60,.60,1}
\definecolor{lightred}{rgb}{1, .60, 0.60}

\newtheorem{theorem}{Theorem}[section]
\newtheorem{lemma}[theorem]{Lemma}
\newtheorem{claim}[theorem]{Claim}

\usepackage{epsfig}
\usepackage{amsmath}
\usepackage{amsfonts}

\newcommand{\nat}{{\bf N}}

\def\nottoobig#1{{\hbox{$\left#1\vcenter
to1.111\ht\strutbox{}\right.\n@space$}}}

\newcommand{\prob}{{\rm Prob}}
\newcommand{\ie}{$\mbox{i.e.}$}

\newlength{\filength}
\newsavebox{\gcbox}
\sbox{\gcbox}{\framebox[\filength]{\rule{0ex}{2ex}}}

\newcommand{\qedblob}{\mbox{\rule[-1.5pt]{5pt}{10.5pt}}}
\def\literalqed{{\ \nolinebreak\hfill\mbox{\qedblob\quad}}}

\def\qed{\literalqed}

\newcommand{\singlespacing}{\let\CS=
\@currsize\renewcommand{\baselinestretch}{1}\tiny\CS}
\newcommand{\singlespacingplus}{\let\CS=
\@currsize\renewcommand{\baselinestretch}{1.25}\tiny\CS}
\newcommand{\doublespacing}{\let\CS=
\@currsize\renewcommand{\baselinestretch}{1.75}\tiny\CS}
\newcommand{\draftspacing}{\let\CS=
\@currsize\renewcommand{\baselinestretch}{2.0}\tiny\CS}

\def\zo{\{0,1\}}

\def\mapping{\rightarrow}

\newcommand{\leqp}{\leq^{+}}
\newcommand{\geqp}{\geq^{+}}
\newcommand{\eqp}{=^{+}}

\newcommand{\cc}[1]{CS^{(#1)}}

\makeatletter
\def\@listI{\leftmargin\leftmargini \parsep 4.5pt plus 1pt minus 1pt\topsep6pt plus 2pt minus 2pt \itemsep  2pt plus 2pt minus 1pt}

\let\@listi\@listI
\@listi
\makeatother

\author{ {Marius Zimand\/}
\thanks{  Department of Computer and Information Sciences, Towson University,
Baltimore, MD. http://triton.towson.edu/\~{ }mzimand ; The author  has been supported in part by the National Science Foundation through grant CCF 1811729.}}
\begin{document}

\title{Secret key agreement from correlated data, with no prior information}

%\date{March 10, 2018}

\date{}

\maketitle

\begin{abstract} 
A fundamental question that has been studied in cryptography and in information theory is whether two parties can communicate confidentially using exclusively  an open channel. We consider the model in which the two parties hold inputs that are correlated in a certain sense. This model has been studied extensively in information theory, and  communication protocols have been designed which exploit the correlation to extract from the inputs  a shared secret key. However, all the existing protocols are not universal in the sense that they require that the two parties also know some attributes of the correlation. In other words, they require that each party knows something about the other party's input. We present a protocol that does not require any prior additional information. It uses space-bounded Kolmogorov complexity to measure correlation and it allows the two legal parties to obtain a common key that looks random to an eavesdropper that observes the communication and is restricted to use a bounded amount of space for the attack. Thus the protocol achieves complexity-theoretical security, but  it does not use any unproven result from computational complexity. On the negative side, the protocol is not efficient in the sense that  the computation of the two  legal parties uses more space than the space allowed to the adversary.
\end{abstract}

%\newpage
%\setcounter{page}{1}
\section{Introduction}

 The goal  of a secret key  agreement protocol is to allow two parties that communicate through a public channel to obtain a shared string that is secret in some reasonable sense (e.g., information-theoretical, complexity-theoretical, or some other sense) to anyone that has observed the communication. 
There are some well-known such protocols, such as the Diffie-Hellman protocol, or various public-key cryptosystems, that are efficient and used in the real world. However, they have the disadvantage of relying on some unproven hardness conjectures
in computational complexity. Another setting is to assume  that the two parties hold at the beginning of the protocol pieces of information that have a certain degree of correlation. Then, in some circumstances, it is possible  to compute the shared secret key without any unproven assumption.  For a simple illustration, suppose that Alice holds a line $L$ in the $2$-dimensional affine space, and Bob holds a point $P$ which lies on $L$. Then Alice  sends Bob the slope of $L$, after which Bob, \emph{knowing that his $P$ is on $L$}, can compute the intercept of $L$. Now, both Alice and Bob have the intercept of $L$, which they can use as a secret key, because the adversary has only seen the slope, which is independent of the intercept.

In this paper,  we consider the latter type of secret key agreement protocols. Thus,  Alice starts with a string $x$, Bob starts with $y$, and, after several rounds of interacting via messages exchanged over a public channel,  they obtain at the end of the protocol  a common secret key, that is a string $z$ which is random conditioned by the transcript of the protocol. The protocol is probabilistically computable, {i.e.},  there exists a probabilistic algorithm so that Alice computes each of her messages by running the algorithm on her input and on the messages that she has received from Bob so far, and Bob computes his messages similarly.  As in the above example, if $x$ and $y$ are \emph{correlated} in some way, one can hope to use the information that is common to these strings to extract with high probability a secret key.

The study of this scenario has a long history in Information Theory and the common flavor of the results is that for many interpretations of   ``correlated,''  secret key agreement is possible. Leung~\cite{leu:t:secret}, Bennett et al.~\cite{ben-bra-rob:j:privacyamplification}, Maurer~\cite{mau:j:seckey},  Ahlswede and Csisz{\'{a}}r~\cite{ahl-csi:j:seckeyone} have started an extensive research line dedicated to the case when $x$ and $y$ are generated by a stochastic process, whose properties describe their correlation (see the survey of Narayan and Tyagi~\cite{nar-tya:b:secrecy}).
 Recently, Romashchenko and Zimand~\cite{rom-zim:c:mutualinfo}  have studied this problem in the very general framework of Algorithmic Information Theory using Kolmogorov complexity to gauge correlation without using any generative model for the provenance of $x$ and $y$. \footnote{We point out that unlike the protocols based on hardness assumption (e.g., Diffie-Hellman protocol) which achieve complexity-theoretic security and are efficient, the protocols in the works above achieve information-theoretic security but do not run in polynomial time.}
 
 In all these works,  Alice and Bob possess at the beginning of the protocol, in addition to  $x$ and $y$,  some information about how these strings are correlated.  For instance, in the above example, Bob knows that the point $P$ is on the line $L$.  In the  scenarios based on generative models, Alice and Bob know various attributes of the joint distribution of the two random variables $(X,Y)$ which describe the stochastic process that generates the pair $(x,y)$, such as entropy, ergodic properties, etc. In the algorithmic information theory setting used in~\cite{rom-zim:c:mutualinfo}, Alice  and Bob know the complexity profile of $(x,y)$, which is the $3$-tuple $(C(x), C(y), C(x,y))$, where $C(\cdot)$ denotes Kolmogorov complexity.  (Throughout this paper,  $C(x)$, called the Kolmogorov complexity of $x$ or the minimal description length of $x$, is the length of a shortest program that when executed by a universal Turing machine prints $x$.)

Can  Alice and Bob agree on a secret key without any additional prior information?  A disclaimer:   This is not really a problem relevant for cryptography, because the protocols are not efficient. We rather view it as  a question about the  fundamental limits of information processing and communication. The challenge is that Alice and Bob have to detect a type of correlation of their inputs through rounds of communication without leaking too much information to the eavesdropper, so that they can compute a shared secret key of reasonable length. 

What is a reasonable length of the secret key? The relevant parameter that comes into play is  the \emph{mutual information} of $x$ and $y$, denoted $I(x:y)$, which intuitively represents the amount of information that is shared by $x$ and $y$. In case we use Kolmogorov complexity to measure the amount of information, $I(x:y)$ is defined as $C(x) + C(y) - C(x,y)$, and,  up to logarithmic precision,  is also equal  to $C(x) - C(x \mid y)$ and to $C(y) - C(y \mid x)$.  It is shown in~\cite{rom-zim:c:mutualinfo} (extending a classical result from~\cite{ahl-csi:j:seckeyone} which is valid for inputs generated by memoryless processes, and which is using Shannon entropy to measure information), that no computable protocol (even probabilistic) can obtain a shared secret key longer than the mutual information of the inputs $x$ and $y$. On the other hand, a protocol is presented in ~\cite{rom-zim:c:mutualinfo} that with high probability produces a shared secret of length $I(x:y)$ (up to logarithmic precison), provided, as mentioned above,  the two parties know the complexity profile of the inputs. 
 Thus, the above discussion suggests that it is natural to aim for a shared secret key whose length is equal to the mutual information of the inputs, for some concept of information that measures the detectable correlation.

\textbf{Our contribution.}  We identify space-bounded Kolmogorov complexity as a concept of information   that allows secret key agreement without any prior information or special setup (e.g., shared randomness, special extra channel) between the two parties.  The space-bounded Kolmogorov  complexity with space bound $S$ of a string $x$, denoted $C^S(x)$,  is similar to standard Kolomogorov complexity except that the universal Turing machine is restricted to use  at most $S$ cells on the working tape (see Section~\ref{s:prereqs} for the formal definition). We show that the correlation  induced by  space-bounded Kolmogorov complexity can be determined without revealing much about $x$ and $y$, which, in turn,  allows the parties to compute a common secret key.

  The protocols that we design produce a  key $z$ that is random given the transcript in the sense of space-bounded Kolmogorov complexity,
where the transcript is the set of messages sent by Alice and Bob.
 Formally, we require that $C^S (z \mid {\rm transcript})$ is close to the length of $z$ (denoted $|z|$), for some space bound $S$. In other words, an eavesdropper which is bounded to use space $S$ and knows the transcript,  needs essentially $|z|$ bits to find the secret key $z$, which is the same as if she did not know the transcript.   If $C^S (z \mid{\rm  transcript}) \geq |z| - \Delta$, we say that $\Delta$ is the randomness deficiency of $z$ with respect to the transcript, and thus we want to obtain $z$ with small randomness deficiency. We also want the length of $z$ to be close to the mutual information of $x$ and $y$, which in the case of space-bounded Kolmogorov complexity is defined as $C^{S_1}(x) - C^{S_2}(x \mid y)$ for space bounds $S_1$ and $S_2$. We next present our results.

We recall that a function $S(n)$ is fully space constructible if there is a Turing machine $M$ that uses exactly $S(n)$ cells for every natural number $n$ and for every input of length $n$. 
\begin{theorem}
\label{t:main0}
Let $S$ be any  fully space constructible function.such that $S(n) \geq n$. 

There is a randomized protocol that allows Alice on input $x$ (an $n$-bit string)  and Bob on input $y$ (of arbitrary length)  to obtain with probability $(1-\epsilon)$  a common string $z$ such that 
\begin{enumerate}
\item[(i)] $|z|  \geqp C^{\lambda_1 \cdot S(n)}(x) - C^{\lambda_2 \cdot S(n)}(x \mid y) $,
\item[(ii)] $C^{S(n)} (z\mid \mbox{transcript}) \geqp  |z| - \Delta$,
\end{enumerate}
where  $\Delta \leq C^{\lambda_3 \cdot S(n)}(x) - C^{\lambda_4 \cdot  S(n)}(x)$, $\lambda_1, \lambda_2, \lambda_3, \lambda_4$ are constants that depend only on the universal Turing machine, and $\geqp$ hides a loss of precision bounded by $O(\log (n/\epsilon))$.

\end{theorem}

Note: The notation $a \geqp b$ means that $a \geq b - \alpha$, where $\alpha$ is the specified loss of precision.
\smallskip 

The communication complexity of the protocol in the above theorem is $n^2 +  O(n\log(1/\epsilon)) + C^{S(n)}(x \mid y)$, which is very large. The protocol in~\cite{rom-zim:c:mutualinfo} (in which the parties also hold the complexity profile of $(x,y)$)  has communication complexity roughly $C(x \mid y)$, which is shown to be optimal. Thus,  in our case, it would be desirable to have a protocol with  communication complexity close to $C^{S(n)}(x \mid y)$. The protocol in the next theorem has \emph{information complexity}  $C^{S(n)}(x \mid y)$ plus a polylogarithmic term and communication complexity  $2C^{S(n)}(x \mid y)$ plus a polylogarithmic term.

\begin{theorem}[Main  Result]
\label{t:main2}
Let $S$ be a fully space constructible function such that $S(n) \geq p_0(n)$, where $p_0(n)$ is a fixed polynomial that only depends on the universal Turing machine.

There is a randomized protocol that allows Alice on input $x$ (an $n$-bit string)  and Bob on input $y$ (of arbitrary length)  to obtain with probability $(1-\epsilon)$  a common string $z$ such that 
\begin{enumerate}
\item[(i)] $|z|  \geqp C^{S(n)}(x) - C^{S(n)}(x \mid y) $,
\item[(ii)] $C^{S(n)} (z\mid \mbox{transcript}) \ge  |z| - \Delta$,
\end{enumerate}
where  $\Delta \leqp C^{\lambda^{-1} S(n)}(x) - C^{\lambda \cdot  S(n)}(x)$, $\lambda$ is a  constant that depends only on the universal Turing machine, and $\geqp$ hides a loss of precision bounded by $O(\log^3(n/\epsilon))$.

\noindent
Furthemore, the length of the transcript is bounded by $2 C^{S(n)}(x \mid y) + O(\log^3(n/\epsilon))$.
\end{theorem}

In the above theorems, the  key $z$ has $\Delta$, the randomness deficiency conditioned by the transcript, bounded by $C^{S'(n)}(x) -  C^{S''(n)}(x)$, where $S'$ and $S''$ differ by a multiplicative constant. Thus, intuitively, $\Delta$ is small.   A particularly  favorable case is when $x$ is a shallow string.  A string $x$ is $S$-shallow if $C^{S(n)}(x) \eqp C(x)$, $\mbox{i.e.}$, if $S(n)$ is enough space to allow the construction of $x$ from a description which is close to a shortest description.  For every space bound $S$, most strings are $S$-shallow and in case $x$ is such a string then  $\Delta  \eqp 0$.

\subsection{Prerequisites}
\label{s:prereqs}

The $S$-space bounded Kolmogorov complexity of $x$ conditioned by $y$ with respect to a Turing machine $M$, denoted $C_M^S (x \mid y)$,  is defined by
\[
C_M^S(x \mid y) = \min \{|p| \mid M(p,y) = x \mbox{ and $M$ uses at most $S$ cells.} \}
\]

 In the case of space-bounded Kolmogorov complexity, simulation by the universal machine  incurs a constant blow-up in space usage.
More precisely, there exists a universal Turing machine $U$ and a constant $\gamma > 1$ such that for any space bound  $S$, for any Turing machine $M$ and for all strings $x, y$,
\[C_U^{\gamma S}(x \mid y) \leq C^{S}_M(x \mid y) + O(1).\]
 As usual,  we fix a universal machine $U$, and denote more simply $C^S(\cdot)$ instead of $C_U^S(\cdot)$. Also, in case the string $y$ used in the condition is the empty string, we drop the condition in  the notation.

The \emph{chain rules} for space-bounded Kolmogorov complexity are as follows: There exists a constant $\gamma > 1$ such that for any space bound $S$, it holds that:
\begin{equation}
\label{e:chain}
\begin{array}{ll}
C^{\gamma S} (x,y) & \leq \quad C^{S}(x) + C^{S}(y \mid x) + O(\log (|x| + |y|)), \\

C^{S} (x,y) &\geq \quad C^{\gamma S} (x) + C^{\gamma S} (y \mid x) +O(\log ( |x| + |y|)).
\end{array}
\end{equation}

To simplify the writing of expressions, we sometimes use the notation $\cc{i}(\dots) $ instead of $C^{\gamma^i \cdot S}(\dots)$, where $S$ is a space bound  and  $\gamma$ (or sometimes $\lambda$) is a constant which is clearly defined in the context. For instance, the last inequality will be written as $\cc{0}(x,y) \geq \cc{1}(x) + \cc{1}(y \mid x) + O(\log(|x| +|y|)$.

\if01
\section{An auxiliary result - probably will be eliminated}

In the next result, Alice and Bob use private random bits but the parameters are not useful if the length of $y$ is much larger (for example, exponentially larger) than the length of $x$.

\begin{theorem}
\label{t:main}
Let $S$ be a fully space constructible function, such that $S(n) \geq p_0(n)$, where $p_0(n)$ is a fixed polynomial that only depends on the universal Turing machine.

There is a randomized protocol that allows Alice on input $x$ (an $n$-bit string)  and Bob on input $y$ (an $m$-bit string)  to obtain with probability $(1-\epsilon)$  a common string $z$ such that 
\begin{enumerate}
\item[(i)] $|z|  \geqp C^{\lambda_1 \cdot S(n+m)}(x) - C^{S(n+m)}(x \mid y) $,
\item[(ii)] $C^{S(n+m)} (z\mid \mbox{transcript}) \geqp |z| - \Delta$,
\end{enumerate}
where  $\Delta \leq C^{S(n+m)}(x) - C^{\lambda_2 \cdot S(n+m)}(x)$,  $\lambda_1, \lambda_2$ are constants that depend only on the universal Turing machine, and $\geqp$ hides a loss of precision bounded by $O(\log(|x|+|y|) + \log 1/\epsilon)$.
\end{theorem}
Note that if the length of $y$ is exponential in the length of $x$, the loss of precision $O(\log(|x|+|y|) + \log(1/\epsilon))$ makes the above expressions useless.  
\fi

\section{Outline of the proofs}

The proofs of both Theorem~\ref{t:main0} and Theorem~\ref{t:main2}     have the same structure. We present an outline, in which, for simplicity, we skip some technical details and ignore small factors in the quantitative relations. Recall that initially Alice holds $x$ and Bob holds $y$. The protocols in both proofs have two phases: (1) \emph{Information reconciliation}, in which Alice communicates  $x$ to Bob by sending him just enough information that allows him to obtain $x$ given his $y$, and (2) \emph{Secret key construction}, in which, separately, Alice and Bob compute the secret key $z$. All the communication happens in the Information reconciliation phase.
\medskip

 \emph{Phase 1} (Information reconciliation):  First, Alice and Bob agree on a space bound $S=S(n)$. Next, Alice sends Bob a randomized hash function $h$.   The goal  is for Alice to send  Bob,  as a fingerprint, some prefix of $h(x)$  that permits  Bob to construct Alice's string $x$ using the fingerprint  and his string $y$.  To avoid sending more information than what Bob needs, Alice sends the bits of $h(x)$ sequentially \emph{one bit per round}.  At each round,  Bob attempts to construct $x$ by checking if the fingerprint of  some string in a set of possible candidates matches the prefix of $h(x)$ that he has received so far. More precisely,  at each round 
$j$, the candidates are those  strings whose $S$-space-bounded complexity conditioned by $y$ is at most $j$.  If  Bob finds a string among these candidates with a fingerprint that matches the prefix of $h(x)$ sent so far by Alice, he believes that he has found $x$, tells  Alice to stop sending further bits by sending her ``1'',  and Phase 1 stops.  Otherwise, he tells Alice that he needs more bits by sending her ``0''  (in which case Alice sends in the  next round the next bit of $h(x)$).

Let $p$ be the prefix of $h(x)$ that Alice sends to Bob  during the entire Phase 1. Then,  with high probability,  at the end of Phase 1, 
 \begin{enumerate}
\item Bob has $x$,
% \item $p$ is a program of $x$ given $y$,
 \item $|p| \leq  C^S (x \mid y)$, because we show that Bob can reconstruct $x$ by round $j = C^{S}(x \mid y)$.
 % \item  $C^{\lambda S}(p \mid x) \eqp 0$, for some constant $\lambda$.  
  In the proof of Theorem~\ref{t:main0},  a random matrix $H$ also appears in the condition (as we explain below), but this has little impact, because $H$ is a random. 
 \end{enumerate}
 \medskip
 
 \emph{Phase 2} (Secret key construction):  After Phase 1,  both Alice and Bob have $x$ (with high probability). They both compute the shared secret key $z$ by exhaustive searching a minimal length
program of $x$ given $p$ in space $S$.
So, from $p$ and $z$, it is possible to construct $x$.  It follows that $z$ and $p$ are independent, because otherwise $z$ would not be minimal.  But then $z$ and the transcript of the protocol  are also almost independent, because the transcript consists of $p$ and the sequence ``$0 \ldots 01$'' sent by Bob,  and the complexity of $0\ldots 01$  is low (at most $\log n +O(1)$).  Thus, $z$ is a secret key. Let us now estimate the length of $z$. Since $x$ can be constructed from $p$ and $z$ in space $S$, it follows that $C^S(x) \leq |p| + |z|$, and thus the length of $z$ is at least $C^S(x) - |p|$, which, by the  above bound of $|p|$, is at least $C^S(x) - C^S(x \mid y)$, which is the mutual information of $x$ and $y$ in the framework of space-bounded Kolmogorov complexity.
\medskip

The main technical issue is finding the hash function that is used in Phase 1.  In the proof of Theorem~\ref{t:main0}, this is just a random linear function given by a random matrix $H$, chosen by Alice. $H$ is roughly $n^2$ bits long, and  Alice needs to also send $H$  to Bob. This is the reason the communication complexity is large.  Also, the information-theoretical considerations in Phase 2, are somewhat more delicate, because we need to take into account $H$. To reduce the communication complexity, one has to use a shorter  hash function. One idea is  to use Newman's theorem from communication complexity, in which $H$ is chosen from a smaller sample space. But the sample space needs to be effectively constructed, and the obvious way to do this leads to a loss of precision that is logarithmic in both the length of $x$ and of $y$, which can be very damaging in case $y$ is much longer than $x$. In Theorem~\ref{t:main2}, we use for hashing an explicit extractor of Raz, Reingold, and Vadhan~\cite{rareva:j:extractor}, which has the special property that if we take prefixes of the output, the extractor property is preserved. These type of extractors, called \emph{prefix extractors}, allow much more communication-efficient hashing,  in the sense that Alice does not need to send the hashing function to Bob, at the cost of making Bob's reconstruction of $x$ more complicated. 

In our technical  approach, we  were inspired by several papers.  Muchnik~\cite{muc:j:condcomp} has introduced  bipartite graphs similar to extractors and has used for a certain type of information reconciliation concepts similar to what we call  \emph{heavy nodes} and \emph{poor nodes} in the proof of Theorem~\ref{t:main2}.  Prefix extractors have been used  for information reconciliation  in~\cite{mus-rom-she:j:muchnik} and~\cite{zim:stoc:kolmslepwolf}, and the first paper analyzes the case of space-bounded Kolmogorov complexity. The application to secret-key agreement is a novel contribution of this paper.  Some of the information-theoretical estimations are similar to those in~\cite{rom-zim:c:mutualinfo}.  The idea of sending pieces of a fingerprint in several rounds for the problem of  information reconciliation (similarly  to our Phase 1) has been used before in~\cite{bra-rao:c;infcomplexity, koz:j:slepwolf}, and, the closest to our approach,  in~\cite{buh-kou-ver:c:randindcomm}, where they study the communication complexity of this problem in terms of the Kolmogorov complexity of the two inputs. There is however a significant difference with the information reconcilation phase in  our main result, because, as standard in communication complexity, the protocol in ~\cite{buh-kou-ver:c:randindcomm} is not computable, and therefore they can use  random hash functions for fingerprinting.

\section{Proof of Theorem~\ref{t:main0}}

%Recall that in  Theorem~\ref{t:main0}, Alice and Bob have access to a public source of randomness.

\textbf{Phase 1: Information reconciliation.} 
%In an information reconciliation protocol, Alice has $x$, Bob has $y$, and the goal is for Alice to communicate $x$ to Bob.
%\medskip

Before sending the first message,  Alice takes  a random  matrix $H$ with entries in the finite field GF[$2$], with $(n+\log(1/\delta))$ rows and $n$ columns, where $n$ is the length of $x$ and $\delta = \epsilon/2n$.  The random matrix $H$ defines a random linear function $h$ mapping $n$ bit strings to $n + \log(1/\delta)$ bit strings (viewed as vectors over GF[$2$]), given by the expression $h(v) = H \cdot v$.

In Round $0$, Alice sends to Bob $n$,  $H$, and  the first $1+\log(1/\delta)$ bits of $h(x)$.
%, and Bob sends to Alice $m$, the length of $y$.

%Next, Alice and Bob use the public source to randomly pick a random linear function $h$ mapping  More precisely,  they pick a binary $(n+\log(1/\epsilon)) \times n$ matrix $H$ chosen uniformly at %random from %the space of all matrices of this type, and the function $h$ is defined by $h(x) = Hx$ (in the field GF($2$)).

Then in each subsequent round, Alice sends to Bob the next bit of $h(x)$ till Bob announces that he does not need any additional bits. Thus, at round $j \ge 1$, Bob has received the first $(j+1) + \log(1/\delta)$ bits of $h(x)$, a string which we denote $p_{j}$. Bob checks if there is a string $u$ in $B_j= \{u \in \zo^{n} \mid C^{S(n)}(u \mid y, H) \leq j\}$ such that $p_{j}$ is a prefix of $h(u)$. If there is such a string $u$, he believes that $u$ is $x$, and announces that he does not need any extra bits and the information reconciliation stops here.  If there is no such string $u$, Bob announces that he needs more bits and the protocol proceeds with the next round.

Bob may be wrong at round $j$, if there is a string $u$ different from $x$ in $B_j$ such that the prefixes of length $(j+1) + \log(1/\delta)$ of $h(u)$ and $h(x)$ coincide.  For an arbitrary string $u \not= x$, the probability that $h(u)$ and $h(x)$ agree in the first $(j+1)  + \log(1/\delta)$ bits is 
$2^{-((j+1) + \log(1/\delta))} = \delta/2^{j+1}$.  Since $B_j$ has less than $2^{j+1}$ elements, by the union bound,  the probability that Bob is wrong at round $j$ is less than $\delta$. 

Let $k = C^{S(n)}(x \mid y, H)$. Let ${\cal E}$ be the event that Bob is wrong at one of the rounds $1, \ldots, k$.  ${\cal E}$ has probability  at most $k \delta \leq (n+c)\delta \leq 2n \delta = \epsilon$. Conditioned by ${\cal E}$ not being true, the protocol reaches round $k$, when Bob finds $x$. Thus, with probability $1-\epsilon$, at the end of round $k$, Bob has obtained $x$, and  the string $p := p_{k}$ is a program for $x$ given $y$ and $H$ in space $\gamma' \cdot S(n)$, for some constant $\gamma'$, and $p$ has length $C^{S(n)}(x \mid y, H)$.
\medskip

\textbf{Phase 2:  Secret key construction.}  By exhaustive search,  Alice and (separately)  Bob find $z$, the first program of $x$ given $p$ and $H$ in space $S$. We show that $z$ satisfies the conclusion of the theorem.
\medskip

We denote  $S:= S(n)$  and we let $\geqp$ hide a loss of precision of $O(\log(n/\epsilon))$. Recall that we use the notation $\cc{i}(\dots)$ in lieu of $C^{\lambda^i \cdot S}(\dots)$, where $\lambda$ is here the maximum between the above $\gamma'$ and $\gamma$ (the constant from the chain rule~\eqref{e:chain}).

First, we notice that, with high probability,  conditioning by a random $H$ does not decrease complexities by too much. 

\begin{claim}
\label{c:rand}
For every space bound ${\cal S}$, for every $n$-bit string $u$, for every string $v$, if $H$ is chosen uniformly at random independent of $u$ and $v$, we have $$\cc{0} (u \mid v,  H) \geqp \cc{2}(u \mid v) \mbox{ with probability $1-\epsilon$}.$$
\end{claim}
\begin{proof} 
\begin{equation}
\label{e:Hestimate}
\begin{array}{ll}
\cc{0} (u \mid v,  H)  &\geqp \cc{1}(u, H \mid v) - \cc{0}(H \mid v) \\
& \geqp \cc{2}(u \mid v)+  \cc{2}(H \mid u,  v) - \cc{0}(H \mid v)  \\
& \geqp \cc{2}(u \mid v) \mbox{ with probability $1-\epsilon$.}
\end{array}
\end{equation}
In the first two lines, we use the chain rule, and in the last line, we use the fact that, for every $i$,  $\cc{i}(H \mid u,v) \geq |H| - \log(1/\epsilon) - 1$, with probability $1-\epsilon$ (by a standard counting argument) and $\cc{i}(H \mid v) \leq |H| + O(1)$ for every $H$.
\end{proof}

Now we  can show  part (i) of Theorem~\ref{t:main0}. 
\begin{equation}
\label{e:zlower}
\begin{array}{lll}
|z| & = \cc{0} (x \mid p, H) 
  & \geqp \cc{1} (x,p \mid H) - \cc{0}(p \mid H) \\ &&\quad\quad\quad  \mbox{(chain rule)} \\
  && \geqp \cc{2}(x \mid H) - |p| \\&& \quad\quad\quad  \mbox{(because $|p| \geqp \cc{0}(p\mid H)$)} \\
  && \geqp \cc{2} (x \mid H) - \cc{0}(x \mid y, H)   \mbox{ with probability $1-\epsilon$} \\&& \quad\quad\quad \mbox{(because $|p| = \cc{0}(x \mid y, H)$)}  \\
   && \geqp \cc{4} (x ) - \cc{0}(x \mid y)   \mbox{ with probability $1-2\epsilon$} \quad \mbox{(by Claim~\ref{c:rand})}
  \end{array}
\end{equation}

Next we move to part (ii), where we need to show that the complexity of the secret key $z$, conditioned by the transcript of the protocol,  is  close to the length of $z$. The transcript consists of $p$, $H$, $n$ (all sent by Alice to Bob) and of Bob's sequence of responses $s = 000 \dots 01$ of length $\ell = k+ 1 + \log (1/\epsilon)$. Bob's sequence has complexity bounded by $\log \ell +O(1) = O(\log(n/\epsilon))$, and therefore, for every $i$ we have $\cc{i}(z \mid s, p, H, n) \eqp \cc{i}(z \mid p, H)$. Thus we can ignore $s$ and $n$ in the condition and it is enough to bound from below  $\cc{i}(z \mid p, H)$. We show the following estimation, which ends the proof of the theorem.

\begin{claim} With probability $1-2 \epsilon$, 
$\cc{4}(z \mid p, H) \geqp |z| - \Delta$, 
where $\Delta = \cc{-2}(x) -  \cc{8}(x)$.
\end{claim}
\begin{proof}
We need an upper bound of $|z|$:
\begin{equation}
\label{e:zupper}
\begin{array}{lll}
|z| & = \cc{0} (x \mid p, H) 
  & \leqp \cc{-1}(x,p \mid H) - \cc{0}(p \mid H) \\ && \quad\quad\quad \mbox{(chain rule)} \\
  && \leqp \cc{-2}(x \mid H) - \cc{1}(x \mid y, H)  \mbox{ with probability $1-\epsilon$}  \vspace{0.1cm}  \\ \vspace{0.1cm}  &&   \hspace{-0.2cm} \mbox{($p$ can be computed from $x$, $H$ and its length;  and $x$ from $p,y,H$)} \\
  && \leqp \cc{-2}(x)  - \cc{3}(x \mid y) \mbox{ with probability $1-2\epsilon$} \\
&&  \quad\quad\quad \mbox{(by Claim~\ref{c:rand})}
\end{array}
\end{equation}
Next, 
\begin{equation}
\label{e:pzupper}
\begin{array}{ll}
\cc{5}(p,z \mid H) & \geqp \cc{6} (x \mid H)   \mbox{ with probability $1-\epsilon$} \\ &\quad\quad \mbox{($x$ can be computed from $p,z,H$)}\\
& \geqp \cc{8}(x)  \mbox{ with probability $1-2\epsilon$}  \quad\quad\quad \mbox{(by Claim~\ref{c:rand})},
\end{array}
\end{equation}
and
\begin{equation}
\label{e:pzlower}
\begin{array}{ll}
\cc{5}(p,z \mid H) & \leqp \cc{4}(p \mid H) + \cc{4}(z \mid p, H) \\   & \quad\quad\quad \mbox{(chain rule)} \\
& \leqp \cc{3}(x \mid y, H) + \cc{4}(z \mid p, H)  \\ &  \quad\quad\quad  \mbox{($p$ can be computed from $x, H$ and its length)} \\
& \leqp \cc{3}(x \mid y) + \cc{4}(z \mid p, H) 
\end{array}
\end{equation}
Combining inequalities~\eqref{e:pzlower} and~\eqref{e:pzupper}, we obtain
\begin{equation}
\cc{4}(z \mid p, H) \geqp \cc{8}(x) - \cc{3}(x \mid y) \mbox{ with probability $1-2\epsilon$.} 
\end{equation}
Using inequality~\eqref{e:zupper}, we finally obtain
\begin{equation}
\cc{4}(z \mid p, H) \geqp |z| - \Delta \mbox{ with probability $1-4\epsilon$,} 
\end{equation}
where $\Delta = \cc{-2}(x) -  \cc{8}(x)$. The conclusion follows after rescaling $\epsilon$.
\end{proof}

\section{Proof of Theorem~\ref{t:main2}}
\medskip

We first present \emph{extractors}, which have been studied in the theory of pseudorandomness (for example, see~\cite{vad:b:pseudorand}). A particular type of extractor, \emph{prefix extractor}, is used in the protocol in the proof of Theorem~\ref{t:main2} for hashing. 

We recall that a $(k, \epsilon)$ extractor is a function $E: \zo^n \times \zo^d \rightarrow  \zo^m$ with the property that for every subset $B \subseteq \zo^n$ of size at least $2^k$ and for every  subset $A \subseteq \zo^m$:
\begin{equation}
 \label{eq:extractorDef}
\bigg| \prob[E(U_B, U_d) \in A] - \frac{|A|}{M} \bigg| < \epsilon,
\end{equation}
where $U_B$ and $U_d$ are independent random variables that are uniformly distributed over $B$ and, respectively, $\zo^d$.

It is useful to view an extractor $E$ as a bipartite graph $G$, whose set of left nodes is $\zo^n$, the set of right nodes is $\zo^m$, and each left node $x$ has $2^d$ (not necessarily distinct) right neighbors  
$\{E(x, w) \mid w \in \zo^d\}.$  The right node $E(x,w)$, for random $w \in \zo^d$, is viewed as the random fingerprint of the left node $x$.  

As usual, we use  \emph{explicit}  extractors. An explicit extractor is a family of extractors $\{E_n\}_{n \in \nat}$ as above, indexed by $n$, and with the rest of the parameters $k,d,m,\epsilon$ being functions of $n$, such that there exists an algorithm that computes $E_n(x,w)$ in time polynomial in $n$. Actually, for us it is more important the space complexity of the algorithm that computes the extractor.

We denote $D = 2^d, M=2^m$. Let $B$ be a set of left nodes. The average numbers of neighbors in $B$ of a right node (called the average $B$-degree) is $avg = |B| \cdot D/M$. We say that a right node $p$ is $\epsilon$-heavy for $B$ if it has more $(1/\epsilon) \cdot  avg$ left neighbors in $B$.  We say that a left node $u \in \zo^n$ is $\epsilon$-poor  for $B$ if a fraction larger than $2\epsilon$ of its right neighbors are $\epsilon$-heavy for $B$. Intuitively, a heavy $p$ is a fingerprint that causes many collisions, and $u$ is poor if many of its fingerprints produce many collisions.

The relevant property of extractors is presented in the next lemma. The point  is that an $\epsilon$-poor string is difficult to  handle because  a random fingerprint of it produces many collisions. The lemma gives a criterion which guarantees that a string is not $\epsilon$-poor.
\begin{lemma}
\label{l:ext}
There exist constants $\lambda > 1$ and $c$ with the following property:

Let $E: \zo^n \times \zo^d \mapping \zo^m$ be a $(k-c, \epsilon)$ extractor computable in space $S(n)$ (in the above sense). 
Let $x$ be  an $n$-bit string (which in the protocol is Alice's input) and $y$ be a string (which is $Bob's$ input), such that $C^{S(n)}(x \mid y, n, k) \leq k$ and $C^{\lambda S(n)} (x \mid y, n, k-1) > k-1$,  and let $B = \{u \in \zo^n \mid C^{S(n)} (u \mid y,n, k) \leq k\}$.  Then $x$ is not $\epsilon$-poor for $B$. 
\end{lemma}
\begin{proof}  Let $A$ be the set of strings that are $\epsilon$-heavy for $B$.  By counting the edges between $B$ and $A$ from left-to-right and from right-to-left, we obtain that $|A|/M \leq \epsilon$.  
Let  ${\rm POOR}$ be the set of nodes that are $\epsilon$-poor for $B$. Note that
 \[
 \prob(E(U_{\rm POOR},U_d) \in A) > 2\epsilon \geq |A|/M + \epsilon, 
\]
It follows that ${\rm POOR}$ has size less than $2^{k-c}$, because otherwise the set ${\rm POOR}$  would violate the property that $E$ is a $(k-c, \epsilon)$-extractor. 

Given $y,n,k,c$, the set ${\rm POOR}$ can be enumerated using space $S(n) +O(n) $ (we need the second term to maintain several counters which require $O(n)$ space). Taking into account the additional space needed by the universal machine, it follows that for some constant $\lambda$, 
if $u$ is an $\epsilon$-poor node then
\[
C^{\lambda S(n) }(u \mid y,n,k,c) \leq k-c+ O(1), 
\]
which implies $C^{\lambda S(n)}(u \mid y,n,k-1) \leq k-1$, for sufficiently large $c$. It follows that $x$ is not $\epsilon$-poor, which proves the lemma.
\end{proof}

We need to use a \emph{prefix extractor}, which is a a function $E: \zo^n \times \zo^d \mapping \zo^n$ with the property that for every $k \leq n$, the function $E_k$ obtained by retaining only the prefix of length $k$ of $E(x,w)$ is a $(k, \epsilon)$ randomness extractor. Raz, Reingold and Vadhan~\cite{rareva:j:extractor} have obtained an explicit extractor $E_{\rm RRV}$  of this type with $d = O(\log^3(n/\epsilon))$. 
 $E_{\rm RRV}(x,w)$ can be computed in time polynomial in $n$ (recall that  $n = |x|$). Let  $p_0 (n)$ be the polynomial that bounds the \emph{space} used in the computation of $E_{\rm RRV}(x,w)$.

In the protocol, we use the Raz-Reingold-Vadhan prefix extractor $E_{\rm RRV}$. We denote by $E_k$, the $k$-prefix of  $E_{\rm RRV}$, and, abusing notation, also the bipartite graph corresponding to the $(k, \epsilon)$ extractor $E_k$. 

In addition to $E_{\rm RRV}$, we use a hash function $h$, based on congruences modulo prime numbers.   We view a string $x$ as an integer (in some canonical way) and define $h_t(x) = (x \bmod{q}, q)$, where $q$ is a prime number chosen at random among the first $t$ prime numbers.  The properties of $h_t$ follow from the following lemma.
\begin{lemma}[\cite{bau-zim:c:linlist}]
\label{l:mod}
Let $x_1, x_2 \ldots, x_s$ be distinct $n$-bit strings, which we view in some canonical way as integers $< 2^{n+1}$. Let $t = (1/\epsilon) \cdot s \cdot n$.  Let $q$ be a prime number chosen uniformly at random among the first $t$ prime numbers. Then, with probability $(1-\epsilon)$,
\[
x_1 \bmod{q} \not\in \{x_2 \bmod{q}, \ldots, x_s \bmod q\}.
\] 
\end{lemma}
\medskip

We now present the protocol. Recall that at the beginning of the protocol, Alice holds an $n$-bit string $x$, and Bob holds a string $y$. We fix the parameters as follows. Let $\lambda$ and $c$ be the constants guaranteed by Lemma~\ref{l:ext}, let $s= (1/\epsilon)\cdot 2^{c+1} \cdot D$,  where $D = 2^d = 2^{O(\log^3(n/\epsilon))}$ is the degree of the $E_{\rm RRV}$ extractor, and let $t = (1/\epsilon) \cdot s \cdot n^2$. We use the space bound $S(n)$ and the constant $\lambda >1$, given by Lemma~\ref{l:ext} applied to the  $E_{\rm RRV}$ extractor.  We assume that the polynomial $p_0$ and the constant $c$, promised by Lemma~\ref{l:ext},  are large enough so that for every string $x$ and every condition string $u$,  $C^{p_0(|x|)}(x \mid u) \leq |x|+c$. As we did earlier, we use the abbreviated notation $\cc{i}(\ldots)$ for $C^{\lambda^i \cdot S(n)}(\ldots)$.
\smallskip

\textbf{Phase 1: Information reconciliation.}

In Round $0$, Alice sends to Bob, $n$ and $h_t(x)$, where $h_t$ is the hash function introduced above.

Next, Alice computes $p' = E_{\rm RRV}(x,w)$ for a random $w \in \zo^d$. 

Alice sends to Bob  the string $p'$ (or rather a prefix of it),  one bit per round, till Bob announces that he does not need more bits.

Suppose we are at round $k$, after Alice has sent the $k$-th bit of $p'$.
Thus, by now Bob has received $p_k$, the $k$-th bit long prefix of $p'$.  He calculates, as we explain next,   a set of candidate strings, which he thinks might be $x$. 
A string $x'$ is a candidate at round $k$ if 
\begin{enumerate}
\item $x'  \in B=  \{u \in \zo^n \mid \cc{n-k} (u \mid y,n,k+c) \leq k+c\}$,  and

\item $x'$ is a neighbor of $p_k$, when viewing $x'$ as a left node  and $p_k$ as a right node in the graph $E_k$, and

\item $x'$ is among the first (in some canonical order) $s$ strings with the above two properties.
\end{enumerate}
 If no candidate has the fingerprint $h_t(x)$, then Bob asks for the next bit of $p'$.  Otherwise, there is one candidate string $x'$ so that $h_t(x') = h_t(x)$.  Then Bob believes that $x'$ is Alice's $x$, and he responds to Alice that he does not need further bits. The Phase 1 (information reconciliation) of the protocol is over.
\medskip

We now analyze Phase 1 (information reconciliation). We show that with high probability, at the end of Phase 1, Bob obtains $x$. 
\medskip

Let $k^* = \min\{ k \mid \cc{n-k} (x \mid y, n, k+c) \leq k+c\}$.  By the above largeness assumptions for $c$ and $p_0(n)$, it follows that $k^* \le n$. Let ${\cal E}$ be the event that there exists $x'$ other than $x$ that is a candidate at one of the rounds  $1, 2, \ldots, k^*$ and has the same fingerprint as $x$ (\ie, $h_t(x') = h_t(x)$).  The total number of candidates from rounds $1, 2, \ldots, k^*$ is at most $k^* \cdot s \leq n \cdot s$. It follows from Lemma~\ref{l:mod}, that  ${\cal E}$ has probability at most $\epsilon$.  Conditioned on ${\cal E}$ not holding, either Bob finds correctly $x$ before round $k^*$ (this happens if $x$ is a candidate at one of these earlier rounds), in which case we are done, or Phase 1 reaches round $k^*$.

Suppose Phase 1 reaches  round $k^*$.  Let $B=  \{u \in \zo^n \mid \cc{n-k^*} (u \mid y,n,k^*+c) \leq k^*+c\}$.   Clearly,  by the definition of $k^*$, $$\cc{n-k^*}  (x \mid y, n, k^*+c) \leq k^*+c$$ and  $$\cc{n-k^* +1}  (x \mid y, n, k^*+c-1) > k^*+c-1.$$   Now we use Lemma~\ref{l:ext} for the pair $(x,y)$, the $(k^*,\epsilon)$ extractor $E_{k^*} : \zo^n \times \zo^d \mapping \zo^{k^*}$ and the set $B$. 
The size of $B$ is less than $2^{k^*+c+1}$ and the average $B$-degree of a right node is $avg = |B|\cdot D/2^{k^*} \leq 2^{c+1} \cdot D$.
 By Lemma~\ref{l:ext} and the two inequalities above, $x$ is not $\epsilon$-poor, which means that with probability $1-2\epsilon$, $p_{k^*}$  is a right neighbor of  $x$ that is not heavy, i.e., it has at most  $(1/\epsilon) \cdot avg \leq  (1/\epsilon) \cdot 2^{c+1} \cdot D= s$ neighbors  in $B$.  Therefore, conditioned on non ${\cal E}$, with probability $1-2 \epsilon$,  $x$ is a candidate at round $k^*$, and Bob finds it. We conclude that with probability larger than $1-3 \epsilon$, Bob correctly obtains $x$.
 
 Let $p$ be the part of the protocol's transcript that Alice has sent to Bob. For the analysis of Phase 2, we need to evaluate the length of $p$. The string $p$ consists of  $n$, $h_t(x)$ and the prefix of $p'$ that Alice has sent bit-by-bit before Bob told her that he does not need any further bits.   By the analysis above, with probability $1-3 \epsilon$, the length of the prefix of $p'$ is at most $k^*$. Let $k = C^{S(n)}(x \mid y) - c$.  Note that $k \leq n$. Since
 $$\cc{n-k} (x \mid y, n, k+c) \leq  \cc{0}(x \mid y) =  k+c,$$ it follows  from the definition of $k^*$ that  $k^* \leq k$. Next, the length of $n$ and $h_t(x)$ is  $O(\log^3(n/\epsilon))$  because the $t$-th largest prime number is less than $t \log t$. We conclude that 
\begin{equation} \label{e:pxy}
|p| \leqp \cc{0}(x \mid y).
\end{equation}  
The communication complexity is $2|p|$,  because it consists of $p$ and of Bob's responses $00\ldots 01$.
\medskip

\textbf{Phase 2: Secret key construction.}

Alice and Bob compute by exhaustive search from $x$ and $p$ a program $z$ of $x$ given $p$ in space $S(n)$ of minimal length $\cc{0}(x \mid p)$.\medskip

We now show that the protocol satisfies the requirements  of Theorem~\ref{t:main2}, and we start with part (i). We let $\geqp$ hide a loss of precision of $O(\log^3(n/\epsilon))$.
We have
\[
\begin{array}{ll}
\cc{0}(x) &\leqp |p| + |z|  \quad\quad\quad \mbox{(because $x$ is computed from $p$ and $z$ in space $S(n)$)}\\ \\
&\leqp \cc{0}(x \mid y) + |z| \quad\quad\quad \mbox{(by \eqref{e:pxy})}
\end{array}
\]
Hence, $|z| \geqp \cc{0}(x) - \cc{0}(x \mid y)$.

Next we show part (ii)  in Theorem~\ref{t:main2}.
First notice that, by the chain rule,
\begin{equation}
|z| = \cc{0}(x \mid p) \leqp \cc{-1}(x,p) - \cc{0}(p).
\end{equation}
Next,  
\[
\begin{array}{ll}
\cc{0} (z \mid p) &\geqp \cc{1}(z,p) - \cc{0}(p) \\ &\quad\quad\quad \mbox{(chain rule)} \\
&\geqp  \cc{1}(x,p) - \cc{0}(p) \\ &\quad\quad\quad \mbox{(because $x$ can be computed from $z$ and $p$ in space $S(n)$)}\\
&= \cc{-1}(x,p) - \cc{0}(p) - (\cc{-1}(x,p) - \cc{1}(x,p)) \\ \\
&\geqp |z| - \Delta,
\end{array}
\]
where $\Delta = \cc{-1}(x,p) - \cc{1}(x,p)$. Since $p$ can be computed from $x$ and the seed of the extractor and the random prime number $q$ used by $h_t$ in space $p_0(n) \leq S(n)$, we have
\[
\begin{array}{ll}
\Delta  &\leqp  \cc{0}(x) - \cc{1} (x). % \\ \\
%&=C^S(x) - C^{\lambda^6 S}(x).
\end{array}
\]
The transcript  of the protocol consists of $p$ and Bob's sequence of responses $00\ldots 01$, which has complexity bounded by $\log n$.
Therefore
\[
\cc{0}(z \mid transcript)  \geqp \cc{0}(z \mid p) \geqp |z| - \Delta, 
\]
which proves part (ii) of Theorem~\ref{t:main2}.~\qed

%\section{Observations regarding the key length}

\section{Final comments}
\label{s:final}
As we have mentioned in the Introduction, the main results are of theoretical, rather than practical,  relevance.  The secret key agreement protocols in Theorem~\ref{t:main0} and Theorem~\ref{t:main2} produce a key that looks random to an adversary whose computation is space-bounded by $S(n)$, and, on the other hand, in both theorems,  the two legal  parties (i.e., Alice and Bob)  execute the protocol in space larger than  $S(n)$.   For this reason, the protocols do not seem to be suitable for real cryptographic applications.

Another observation regards the key length.  In Theorem~\ref{t:main2}, the protocol, on inputs the $n$-bit string $x$ and the string $y$,  runs in space bounded by $\lambda^n  S(n)$ (we take into account the space used by the two parties combined) for some constant $\lambda > 1$  and produces a secret key $z$ of length $|z| \approx C^{S(n)}(x) - C^{S(n)} (x \mid y)$  and having the randomness deficiency of $z$ conditioned by the trancript  as stated in the theorem. Recall that the randomness deficiency $\Delta$ is defined by $\Delta =  |z| - C^{S(n)}(z \mid transcript)$.  Is the length of $z$ optimal?   It is known from~\cite{rom-zim:c:mutualinfo}, that no computable protocol can produce a key longer than $C(x) - C(x \mid y)$, the mutual information of the inputs hold by the two parties. We have not been able to obtain a similarly clean result for protocols that run in space $S(n)$. By adapting the arguments in~\cite{rom-zim:c:mutualinfo}, it can be shown, that  if a protocol  runs in space $S(n)$ then, for every pair of inputs $(x,y)$ with length bounded by $n$,  it produces a key $z$ with $C^{\lambda^2 S(n)} (z \mid transcript) \leq C^{S(n)}(x) - C^{\lambda^3 S(n)} (x \mid y)$, for some constant $\lambda > 1$. Thus we obtain the following upper bound: If a secret key agreement protocol runs in space $S(n)$  and on input $(x,y)$, with $|x|,|y| \leq n$, it produces a secret key $z$ with 
 randomness deficiency $\Delta$, then
\[|z| \leq   C^{S(n)}(x) - C^{\lambda^3 S(n)} (x \mid y) + \Delta + \Delta_1,\]
 where $\Delta_1 = C^{(S(n)}(z \mid transcript) - C^{\lambda^2 S(n)} (z \mid transcript)$ and $\lambda > 1$ is a constant.
 %\marius{[To check last statement; does $\lambda$ depend on the protocol?}

\section{Acknowledgements}
I want to thank  Andrei Romashchenko for useful discussions. I also thank the anonymous referees for their observations which have helped me correct some errors and improve the presentation.

\bibliography{theory-3}
\bibliographystyle{alpha}

\end{document}